\documentclass{llncs}

\usepackage[utf8]{inputenc}
\usepackage{enumerate}
\usepackage{amsmath}
\usepackage{amssymb}
\usepackage{listings}
\usepackage{xspace}
\usepackage{vaucanson-g}

\DeclareMathOperator{\rk}{rk}

\newcommand{\Adfas}{\Adfa{}s\xspace}
\newcommand{\Adfa}{\textsc{ADFA}\xspace}
\newcommand{\Madfa}{\textsc{MADFA}\xspace}
\newcommand{\Idfa}{\textsc{IDFA}\xspace}
\newcommand{\Icdfa}{\textsc{ICDFA}\xspace}

\newcommand{\Dfa}{\textsc{DFA}}

\newcommand{\bdisplay}{\begin{description}\footnotesize\item[]}
\newcommand{\edisplay}{\end{description}}

\newcommand{\bquot}[1]{\begin{quotation}\small\noindent
  \textbf{#1}\hspace{\labelsep}\ignorespaces}
\newcommand{\equot}{\unskip\end{quotation}}

\begin{document}%

\title{Exact Generation of  Acyclic Deterministic Finite Automata
\thanks{This work was partially funded by Funda\c{c}\~ao para a
  Ci\^encia e Tecnologia (FCT) and Program POSI, and by project ASA
  (PTDC/MAT/65481/2006). }
\thanks{This paper was presented at the 10th Workshop on
  Descriptional Complexity of Formal Systems, DCFS'08}}
\author{
Marco Almeida\thanks{Marco Almeida is funded by FCT grant
  SFRH/BD/27726/2006.}
\hspace{0.3cm}  Nelma Moreira \\
 \hspace{0.3cm} Rog\'erio Reis\\
{\tt \{mfa,nam,rvr\}@ncc.up.pt}}
\institute{DCC-FC \ \& LIACC,  Universidade do Porto \\
  R. do Campo Alegre 1021/1055, 4169-007 Porto, Portugal }

\date{}
\maketitle

\lstdefinelanguage{algo}{
  morekeywords={for,if,then,do,while,return,break,else,continue,and,nil,to,
downto,print,def,output},
  morecomment=[l]{\%}}

\lstset{language=algo,
  aboveskip=20pt, belowskip=20pt, %numbers=left, numberstyle=\tiny,
  mathescape=true , basicstyle=\small,
  xleftmargin=20pt,
  xrightmargin=20pt,escapechar=@}

\begin{abstract}

  We give a canonical representation for trim acyclic deterministic
  finite automata (\Adfa{}) with $n$ states over an alphabet of $k$
  symbols. Using this normal form, we present a backtracking algorithm
  for the exact generation of \Adfa{}s. This algorithm is a non
  trivial adaptation of the algorithm for the exact generation of
  minimal acyclic deterministic finite automata (\Madfa{}), presented
  by Almeida \textit{et al.}.
\end{abstract}

\section{Introduction}
\label{sec:intro}
Recently, Liskovets~\cite{liskovets06:_exact} obtained a formula for
the enumeration of unlabelled (non-isomorphic) initially connected
acyclic deterministic finite automata with $n$ states over an alphabet
of $k$ symbols.
Callan~\cite{callan07:_stirl_cycle_number_count_unlal} presented a
canonical form for those automata and showed that a certain
determinant of Stirling cycle numbers can also count them. That
canonical form is obtained by observing that if we mark the visited
states, starting with the initial state, it is always possible
to find a state whose only incident states are already marked.  This
induces a unique labelling of states, but it is not clear how these
representations can be used in automata generation.  Almeida \textit{et
al.}~\cite{almeida07:_exact} obtained a canonical form for
(non-isomorphic) minimal acyclic deterministic finite automata
(\Madfa{}) and an exact generation algorithm. Unfortunely the
canonical form did not provide directly an enumeration formula for
\Madfa{}s. One of the applications of such an enumeration formula would
be in the development of uniform random generators of automata,
useful for the average case analysis of algorithms for that
class of automata.
The enumeration of different kinds of finite automata was considered
by several authors since late 1950s. For more complete surveys we
refer the reader to Domaratzki~{\textit et al.}~\cite{domaratzki02} and to
Domaratzki~\cite{domaratzki06:_enumer_formal_languag}.
Liskovets~\cite{liskovets69} and Robinson~\cite{robinson85:_count}
counted non-isomorphic initially connected deterministic finite
automata (\Icdfa).  More recently, several authors examined related
problems. Domaratzki {\textit et al.}~\cite{domaratzki02} studied the
(exact and asymptotic) enumeration of distinct languages accepted by
finite automata with $n$ states. Nicaud~\cite{nicaud00}, Champarnaud
and Parantho\"en~\cite{champarnaud:_random_gener_dfas} presented a
method for randomly generating complete \Icdfa's. Bassino and
Nicaud~\cite{bassino07:_theor_comput_scien} showed that the number of
complete \Icdfa's is $\Theta(n2^nS(kn,n))$, where $S(kn,n)$ is a
Stirling number of the second kind.  Based on a canonical string
representation for \Icdfa's, Almeida {\textit et
  al.}~\cite{almeida07:_enumer_gener_strin_autom_repres} obtained a
new formula for the number of non-isomorphic \Icdfa's, and provided
exact and uniform random generators for them.

In this paper, we give a canonical representation for trim (complete)
acyclic deterministic finite automata (\Adfa{}). By trim we mean that
from the initial state all other states are reachable (initially
connected)  and that from all states (but the dead) at least one final
state is reachable (useful).

This canonical form extends the one for \Madfa{}s 
by taking into consideration equivalent states. The backtracking algorithm for the
exact generation of \Adfa{}s is a non-trivial adaptation of the one for
\Madfa{}s, because we must properly consider the equivalence classes
but still avoid the multiple generation of isomorphic automata.
It is easy to order equivalent states according to the
words that reach them (i.e., their left languages) but to obtain a
feasible generator algorithm  we must find an
ordering such that:
\begin{itemize}
\item the canonical representation for \Adfas is a natural
  extension of the canonical representation for \Madfa{}s (i.e.,
  preserves its characteristics);
\item it allows the detection of an ill-formed automata representation as soon
  as possible (as the algorithm proceeds  backwards, towards the initial
  state);
\item it allows the  exact generation algorithm to output the automata canonical
  representations in increasing order.
\end{itemize}
\Adfa{}s, as defined here, are a proper subset of the class of acyclic
automata enumerated by Liskovets and Callan because we only consider
automata where all the states are useful. Once more, their formulae
can not be used directly, but in this paper we hope to contribute to
a better understanding of the internal structure of~\Adfa{}s. 

The paper is organized as follows. In the next section some basic
concepts and notation are introduced. In Section~\ref{sec:adfa} we
review some concepts about acyclic deterministic finite automata and
the canonical form for \Madfa{}s.  In Section~\ref{sec:normal} we show
how to extend that canonical form to \Adfa{}s. In
Section~\ref{sec:exact} we describe an algorithm to efficiently
generate equivalent states as an extension to the exact generator for
\Madfa{}s.  Some experimental results are also summarized in that
section.  In Section~\ref{sec:count} we consider \Adfa{}s enumeration
formulas for small values of $n$ and $k$. Finally
Section~\ref{sec:con} concludes.

\section{Basic concepts}
\label{sec:basic}
We review some basic concepts we need in this paper.
For more details we refer the reader~to~Hopcroft {\itshape et
al.}~\cite{hopcroft00_c:_introd_autom_theor_languag_comput}, 
Yu~\cite{yu97:_handb_formal_languag} or Lothaire~\cite{lothaire05:_algor_words}.

Let $[n,m]$ denote the set $\{i\in \mathbb{Z}\mid n\leq i \leq m\}$.
In a similar way, we consider the variants $]n,m]$, $[n,m[$ and
$]n,m[$.  Whenever we have a finite ordered set $A$, and a function
$f$ on $A$, the expression  $(f(a))_{a\in A}$ denote the values of $f$
for increasing values of $A$.

Let  $\Sigma$  be an {\itshape alphabet} and $\Sigma^\star$ be the set of
all words over $\Sigma$. The empty word
is denoted by $\varepsilon$.  The length of a word
$x=\sigma_1\sigma_2\cdots\sigma_n$, denoted by $|x|$, is $n$. 
 A language $L$
is a subset of $\Sigma^\star$. A language is finite if its cardinality
is finite.

The alphabet $\Sigma$ can be equipped with a total order $<$ that
allows the definition of total orders on $\Sigma^\star$. A \textit{lexicographical}
order on $\Sigma^\star$  is defined as follows. 
Let $x=x_1\ldots x_m,\, y=y_1\ldots y_n \in \Sigma^\star$.  Then $x<y$ if:
\begin{enumerate}
\item there exists an integer $j \in[1,\min\{m,n\}]$ such that
  $(\forall i\in [1,j[)\, x_i=y_i$ and $x_j<y_j$;
\item $m<n$ and $(\forall i\in [1,m])\,x_i=y_i$.
\end{enumerate}

A {\itshape deterministic
  finite automaton} (\Dfa{}) ${\cal A}$ is a tuple
$(S,\Sigma,\delta,s_0,F)$ where $S$ is a finite set of states,
$\Sigma$ is the alphabet, $\delta: S \times \Sigma \rightarrow S$ is
the transition function, $s_0$ the initial state and $F\subseteq S$
the set of final states. 

We assume that the transition function is total, so we consider
  only {\itshape complete} \Dfa{}s. 
  The transition function $\delta$ is
  inductively extended to $\Sigma^\star$, by $(\forall s\in
  S)\,\delta(s,\varepsilon)=s$ and
  $\delta(s,x\sigma)=\delta(\delta(s,x),\sigma)$.

A \Dfa{} is {\itshape initially connected} (or {\itshape accessible})
(\Icdfa{}) if for each state $s\in S$ there exists a word~$x\in
\Sigma^{\star}$ such that $\delta(s_0,x)=s$.  A \Dfa{} is {\itshape
  trim} if it is an \Icdfa{} and every state is {\itshape useful},
i.e., $(\forall s\in S)(\exists x\in \Sigma^\star)\,\delta(s,x)\in F$.

Two \Dfa{}s $(S,\Sigma,\delta,s_0,F)$ and
$(S',\Sigma',\delta',s_0',F')$ are called {\itshape isomorphic} if
$|\Sigma|=|\Sigma'|=k$, there exist bijections
$\Pi_1:\Sigma\rightarrow[0,k-1]$, $\Pi_2:\Sigma'\rightarrow[0,k-1]$
and a bijection $\iota:S \rightarrow S'$ such that $\iota(s_0)=s_0'$,
$\iota(F)=F'$, and
for all $\sigma\in \Sigma$ and $s\in S$,
$\iota(\delta(s,\sigma))=\delta'(\iota(s),\Pi_2^{-1}(\Pi_1(\sigma)))$.

The {\itshape language} accepted by a \Dfa{} ${\cal A}$ is ${\cal L}({\cal
  A})=\{x\in \Sigma^\star\mid \delta(s_0,x)\in F\}$. For a state $s\in
S$ we denote
\begin{align*}
  {\cal L}_L({\cal  A},s)&=\{x\in\Sigma^\star \mid  \delta(s_0,x)=s\},\\
  {\cal L}_R({\cal  A},s)&=\{x\in\Sigma^\star \mid \delta(s,x) \in F\},
\end{align*}
\noindent the \textit{left} and the \textit{right language} of state
$s$, respectively. We omit ${\cal  A}$ whenever no confusion
arises. All states of a \Dfa{} have distinct \textit{left} languages.

 Two \Dfa{}s are
{\itshape equivalent} if they accept the same language. 
We say that two states $s$ and $s'$ are {\itshape equivalent},
$s\sim s'$, if
and only if ${\cal L}_R({\cal  A},s)={\cal L}_R({\cal  A},s')$.
%$$(\forall x \in \Sigma^\star) \,(\delta(s,x)\in F \leftrightarrow \delta(s',x)\in F).$$
A \Dfa{} is {\itshape minimal} if it has no equivalent states and
it is initially-connected.  Minimal \Dfa{}s are unique up to isomorphism.

\section{Acyclic finite automata}
\label{sec:adfa}
An {\itshape acyclic deterministic finite automaton} is a \Dfa{}
${\cal A}=(S \cup\{\Omega\},\Sigma,\delta,s_0,F)$ with $F\subseteq S$
and $s_0\not=\Omega$ such that $(\forall \sigma \in \Sigma)\,
\delta(\Omega,\sigma) = \Omega$ and $(\forall x \in
\Sigma^\star\setminus \{\epsilon\})(\forall s \in S)\,\delta(s,x)\not= s$. The state
$\Omega$ is called the {\itshape dead state}, and is the only
{\itshape cyclic} state of ${\cal A}$. The {\itshape size of} $\mathcal{A}$ is
$|S|$.  We are going to consider only trim complete acyclic
deterministic finite automata (\Adfa{}), where
all states but $\Omega$ are useful. It is obvious that the language of
an \Adfa{} is finite.

A state $s\in S$ is called {\itshape pre-dead} if $(\forall \sigma\in
\Sigma)\,\delta(s,\sigma)=\Omega$. Every \Adfa{} has at least a
pre-dead state and all pre-dead states are final.

Given an \Adfa{}, ${\cal A}=(S\cup \{\Omega\},\Sigma,\delta,s_0,F)$,
the {\itshape rank} of a state $s\in S$, denoted $\rk(s)$, is the
length of the longest word $x\in \Sigma^\star$ such that
$\delta(s,x)\in F$ (i.e., $x \in {\cal L}_R({\cal A},s)$). The {\itshape
  rank}\footnote{Also called the {\itshape diameter} of ${\cal A}$.}
of an $\Adfa$~${\cal A}$ , $\rk({\cal A})$, is $\max\{\rk(s)\mid s\in
S\}$.  Trivially, we have that $\rk(s_0)=\rk({\cal A})$ and
$\rk(s)=0$, for all pre-dead states $s$.

%  Given an \Adfa{} the rank of
% each state can be determined by the following algorithm:
%
% \begin{lstlisting}
% for $s$ in S
%    $\rk(s)$ $\leftarrow$ $\bot$
% rank($s_0$)
% def rank($s$)
%   if $\rk(s)$ $\not= \bot$ then return $\rk(s)$
%   r $\leftarrow$ 0
%   for $\sigma\in \Sigma$
%     if $\delta(s,\sigma) \not= \Omega$ then r $\leftarrow$ max(r,1+rank($\delta(s,\sigma)$))
%   $\rk(s)$ $\leftarrow$ r
%   return r 
% \end{lstlisting}

For every state $s\in S$, with $\rk(s)>0$ there exists a transition to
a state with rank  immediately lower than  $s$'s.
\begin{lemma}\label{lem:rkmu}
  Let   ${\cal A}=(S\cup \{\Omega\},\Sigma,\delta,s_0,F)$ be an
  \Adfa, then 
$$(\forall s \in S) (\rk(s)\not=0 \Rightarrow 
(\exists \sigma \in \Sigma)\; \rk(\delta(s,\sigma))=\rk(s)-1).$$
\end{lemma}
Two states $s$ and $s'$ are \textit{mergeable}   if they are both either
final or not final, and the transition function is identical, i.e.,
$$ (s \in F\;\leftrightarrow\; s' \in F) \wedge
(\forall \sigma\in \Sigma)\; \delta(s,\sigma)=\delta(s',\sigma).$$

For instance, in the \Adfa of Figure~\ref{fig:adfa} the states $s_2$
and $s_3$ are mergeable, and $s_7$ and $s_8$ are mergeable too.
\begin{figure}[htb]
  \centering
      \VCDraw{%
        \TinyPicture
        \begin{VCPicture}{(0,-6.3)(16,3.3)}
          \State[s_0]{(0,0)}{S0}
          \Initial{S0}
          \State[s_1]{(5,3)}{S1}
          \State[s_2]{(3.5,0)}{S2}
          \State[s_3]{(2,-3)}{S3}
          \State[s_4]{(9,3)}{S4}
          \State[s_5]{(9,0)}{S5}
          \State[s_6]{(12,0)}{S6}
          \FinalState[s_7]{(16,0)}{S7}
          \FinalState[s_8]{(9,-2)}{S8}
          \State[\Omega]{(9,-4.3)}{Omega}
          \EdgeL{S0}{S1}{a} \EdgeL{S0}{S2}{b} \EdgeL{S0}{S3}{c}
          \EdgeL{S1}{S4}{a,b} \EdgeL{S1}{Omega}{c}
          \EdgeL{S2}{S5}{a} \EdgeL{S2}{Omega}{b,c}
          \EdgeL{S3}{S5}{a} \EdgeL{S3}{Omega}{b,c}
          \EdgeL{S4}{S6}{a} \EdgeL{S4}{S7}{b,c}
          \EdgeL{S5}{S4}{b} \EdgeL{S5}{S6}{a,c}
          \EdgeL{S6}{S7}{a,b} \EdgeL{S6}{S8}{c}
          \EdgeL{S7}{Omega}{a,b,c}
          \EdgeL{S8}{Omega}{a,b,c}
          \LoopS{Omega}{a,b,c}
        \end{VCPicture}
      }
  \caption{An \Adfa.}
  \label{fig:adfa}
\end{figure}

An \Adfa{} can be minimized by merging mergeable states, thus, a
minimal \Adfa{} (\Madfa{}) can be characterized by:

\begin{lemma}[\cite{revuz92:_minim_of_acycl_deter_autom,lothaire05:_algor_words}]\label{lem:madfa}
  An \Adfa{} ${\cal A}=(S\cup \{\Omega\},\Sigma,\delta,s_0,F)$ is
  minimal if and only if it has no mergeable states.
%$(\forall s,s'\in S\cup \{\Omega\})((s \in F\; \dot{\vee}\; s' \in F) \vee
%(\exists \sigma\in \Sigma)\; \delta(s,\sigma)\not=\delta(s',\sigma)).$
\end{lemma}

It is a direct consequence of Lemma~\ref{lem:madfa} that every
\Madfa{} has  a unique pre-dead state, $\pi\in S$, and that mergeable
states have the same rank. This implies that to minimize an \Adfa{} 
it is only necessary to merge states by increasing rank order (see
Revuz~\cite{revuz92:_minim_of_acycl_deter_autom} or Lothaire~\cite{lothaire05:_algor_words}).

\subsection{A normal form for \Madfa{}s}
\label{sec:nfmadfas}
Based upon the above considerations, Almeida \textit{et
  al.}~\cite{almeida07:_exact} presented a canonical representation
for \Madfa{}s.

Let ${\cal A}=(S\cup \{\Omega\},\Sigma,\delta,s_0,F)$ be a \Madfa{}
with $k=|\Sigma|$ and $n=|S|\geq 2$. Consider a total order
over~$\Sigma$ and let $\Pi:\Sigma\longrightarrow[0,k[$ be the
bijection induced by that order. Let $R_l=\{s\in S\mid \rk(s)=l\}$.
It is possible to obtain a canonical numbering of the states
$\varphi:S\cup\{\Omega\} \rightarrow [0,n]$ proceeding by increasing
rank order and considering an ordering over the ($k+1$)-tuples that
represent the transition function and the finality of each state.
For each state $s\in S$, let its \textit{representation} be a
$(k+1)$-tuple $\Delta(s) = (\varphi(\delta(s,\Pi^{-1}(0))),\ldots,
\varphi(\delta(s,\Pi^{-1}(k-1))),f)$, where the first $k$ values
represent the transitions from state $s$ and the last value, $f$, is
$1$ if $s\in F$ or $0$, otherwise. 
%If the last value is omitted we
%denote the representation by $\hat{\Delta}(s)$.
  Let
$\varphi(\Omega)=0$ and $\varphi(\pi)=1$. Thus, the representations of
$\Omega$ and $\pi$ are $(0^k,0)$, and $(0^k,1)$, respectively. We can
continue this process considering the states by increasing rank order,
and in each rank we number the states by lexicographic order over
their transition representations.  It is important to note that
transitions from a given state can only refer to states of a lower
rank, and thus already numbered. The sequence of tuples
$(\Delta(i))_{i\in[0,n]}$ is the \textit{canonical string
  representation} of ${\cal A}$.  Formally, the assignment of state numbers, $\varphi$, can be
described by the following simple algorithm\label{phi}:
\begin{lstlisting}
$\varphi(\Omega)\leftarrow 0$;$\varphi(\pi)\leftarrow 1$;$i \leftarrow 2$
for $l$ in $] 0,\rk({\cal A})]$
  for $s\in R_l$ by lexicographic order over $\Delta(s)$
    $\varphi(s) \leftarrow i$
    $i \leftarrow i + 1$
\end{lstlisting}
In Figure~\ref{fig:ex_MADFA}, we present a \Madfa{}  
($n=7$ and $k=3$), the $\varphi$ function and its canonical
representation.
\begin{figure}[htb]
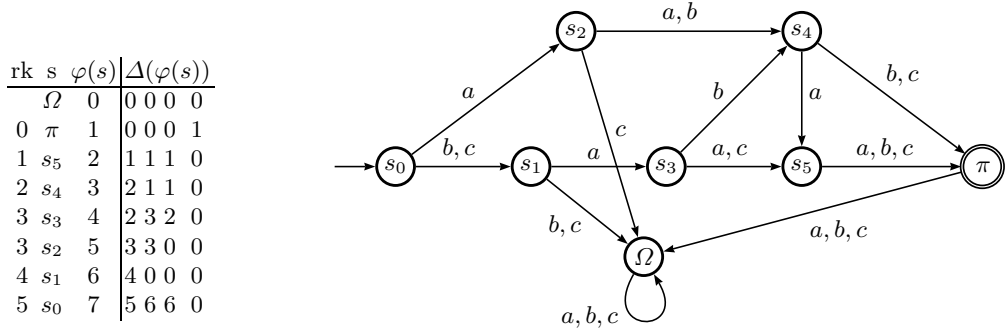

\begin{tabular}[b]{lcl}
  \begin{minipage}{4cm}
  \begin{tabular}[c]{ccc|cccccc}
    $\rk$&s&$\varphi(s)$&\multicolumn{4}{c}{$\Delta(\varphi(s))$}%tuple}
    
\\\hline
       &$\Omega$&$0$&$0$&$0$&$0$&$0$\\
    $0$&$\pi$&$1$&$0$&$0$&$0$&$1$\\
    $1$&$s_5$&$2$&$1$&$1$&$1$&$0$\\
    $2$&$s_4$&$3$&$2$&$1$&$1$&$0$\\
    $3$&$s_3$&$4$&$2$&$3$&$2$&$0$\\
    $3$&$s_2$&$5$&$3$&$3$&$0$&$0$\\
    $4$&$s_1$&$6$&$4$&$0$&$0$&$0$\\
    $5$&$s_0$&$7$&$5$&$6$&$6$&$0$
  \end{tabular}
  \end{minipage}&&
  \begin{minipage}{5cm}
  \VCDraw{
    \TinyPicture
            \begin{VCPicture}{(-1.5,-5)(16,4)}
              \State[s_0]{(0,0)}{S0}
              \Initial{S0}
              \State[s_1]{(3,0)}{S1}
              \State[s_2]{(4,3)}{S2}
              \State[s_3]{(6,0)}{S3}
              \State[s_4]{(9,3)}{S4}
              \State[s_5]{(9,0)}{S5}
              \FinalState[\pi]{(13,0)}{S6}
              \State[\Omega]{(5.5,-2)}{Omega}
              \LoopS{Omega}{a,b,c}
              \EdgeL{S0}{S1}{b,c} \EdgeL{S0}{S2}{a}
              \EdgeL{S1}{S3}{a} \EdgeR{S1}{Omega}{b,c}
              \EdgeL{S2}{Omega}{c} \EdgeL{S2}{S4}{a,b}
              \EdgeL{S3}{S4}{b} \EdgeL{S3}{S5}{a,c}
              \EdgeL{S4}{S5}{a} \EdgeL{S4}{S6}{b,c}
              \EdgeL{S5}{S6}{a,b,c}
              \EdgeL{S6}{Omega}{a,b,c}
            \end{VCPicture}  }
  \end{minipage}
\end{tabular}
  \caption{An example of a \Madfa{} that can be described by the canonical
    representation 
    $[[0,0,0,0],[0,0,0,1],[1,1,1,0],[2,1,1,0],[2,3,2,0],[3,3,0,0],[4,0,0,0],[5,6,6,0]]$.}\label{fig:ex_MADFA}
\end{figure}

The characterization of these strings and that they constitute a
canonical representation for \Madfa{}s is given by the following
theorem:
\begin{theorem}[{\cite[Thms.3-5.]{almeida07:_exact}}]\label{th:mcr}
  There exists a bijection between non-isomorphic \Madfa{}s with $n$
  states and $k$ symbols and the set of strings $(s_i)_{i \in
    [0,(k+1)(n+1)[}$, with $s_i\in [0,n[$ that satisfy the following
  conditions. Let $(f_i)_{i\in [1,n[}$ be the sequence of the
  positions in $(s_i)_i$ of the first occurrence of each
  $i\in[1,n[$. Let $d\leq n$ and let $(r_l)_{l\in [0,d]}\in[1,n]$ be
  the sequence of the first states of each rank in $(s_i)_i$. Then:
\begin{gather*}
  s_0=\cdots=s_k=\cdots=s_{2k}=0\wedge s_{2k+1}=1\tag{N0}\label{eq:N0}\\
  (\forall i\in[0,n])\,s_{(k+1)i+k}\in\{0,1\}\tag{N1}\label{eq:N1}\\
  r_0=1\wedge r_1=2\wedge r_{d}=n\wedge (\forall l\in[0,d[)\,r_l<r_{l+1}
  \tag{N2}\label{eq:N2}\\
  \begin{split}
  ((\forall i\in [1,n[)&\,s_{f_i}=i\;\wedge \\(\forall j\in[0,n])&(\forall
  m\in[0,k[)\,((k+1)j+m <f_i\;\Rightarrow\; s_{(k+1)j+m}\neq i)) 
  \end{split}\tag{N3}\label{eq:N3}\\
  (\forall l\in[0,d[)(\forall i\in[r_l,r_{l+1}[)\, kr_{l+1}+1\leq
  f_i \tag{N4}\label{eq:N4}\\
  (\forall l\in[0,d])(\forall i\in[r_l,r_{l+1}[)(\exists
  m\in[0,k[)\; s_{(k+1)i+m}\in[r_{l-1},r_l[ \tag{N5}\label{eq:N5}\\
  (\forall l\in[0,d[)(\forall i\in[r_l,r_{l+1}-1[)\,
  (s_{(k+1)i+m})_{m\in[0,k]} < (s_{(k+1)(i+1)+m})_{m\in[0,k]} 
  \tag{N6}\label{eq:N6}
\end{gather*}
\end{theorem}
 
The condition~\ref{eq:N0} gives the representation of the dead ($\Omega$)
and the pre-dead state ($\pi$). The condition~\ref{eq:N1} states that
the last symbol of each state representation indicates if the state is
final or not.  The condition~\ref{eq:N2} ensures that states are
numbered by increasing rank order. The condition~\ref{eq:N3} defines
the sequence $(f_i)_{i\in[1,n[}$, and ensures that ${\cal A}$ is
initially connected. The condition~\ref{eq:N4} is a direct consequence
of the rank definition, i.e., a state can only refer to a state of a
lower rank. The condition~\ref{eq:N5} states that every state has a
transition to a state with rank immediately lower than its own. The
condition~\ref{eq:N6} ensures that within a rank the state
representations are lexicographically ordered.

\section{A normal form for \Adfa{}s}
\label{sec:normal}

If an \Adfa{} is not minimal, then it has at least two mergeable
states, but not all equivalent states need to be mergeable. The two following lemmas give characterizations of equivalent states in an
\Adfa{} that will be used to obtain a  canonical representation.

\begin{lemma}\label{lem:equivrank}
In an \Adfa every two  equivalent states must belong to the same rank.
\end{lemma}
This follows directly from the definitions.

\begin{lemma}\label{lem:equiv}
Let   ${\cal A}=(S\cup \{\Omega\},\Sigma,\delta,s_0,F)$ be an
  \Adfa.  For all $s,s'\in S$, if $s\sim s'$ then there exists $w\in
  \Sigma^\star$, such that $\delta(s,w)$ and $\delta(s',w)$ are mergeable states.
\end{lemma}
\begin{proof}
  If $s\sim s'$ then $(\forall \sigma\in \Sigma),\delta(s,\sigma)\sim
  \delta(s',\sigma)$. Suppose that there exists $\sigma_1\in \Sigma$ such that
  $s_1=\delta(s,\sigma_1) \not=  \delta(s',\sigma_1)=s_1'$. Because
  $s_1\sim s_1'$ we can proceed as before, but because ${\cal A}$ is
  acyclic and $|S|$ is finite this process must stop, and two
  mergeable states, $s_j$ and $s_j'$ for $j\leq |S|$, must be reached. 
The concatenation of the $\sigma_1\ldots\sigma_j$ provides the word $w$.
\end{proof}

In order to have a canonical representation for \Adfa{}s we must
provide an ordering for the equivalent states. Because they must
appear in the same rank we may restrict the state ordering by rank and
consider a proper extension of the function $\varphi$ (assignment of
state numbers), and so a proper extension of the canonical
representation for \Madfa{}s. In particular we take
$\varphi(\Omega)=0$. Because \Adfa{}s are deterministic, we have
\begin{lemma}
  Let   ${\cal A}=(S\cup \{\Omega\},\Sigma,\delta,s_0,F)$ be an
  \Adfa. Then $$(\forall s\not=s'\in S\cup \{\Omega\}), {\cal L}_L(s)\cap
  {\cal L}_L(s')=\emptyset.$$
\end{lemma}
Any  two different states can be distinguished, if we define any
injective function $\Psi: S \rightarrow O$,
where $O$ must be  a total ordered set. 
%We denote by $\Psi(s)$ the \textit{characteristic word} of $s$. 

For instance, given an order over $\Sigma$  we could have
$\Psi: S \rightarrow \Sigma^\star$ given by
$\Psi(s)=\min\{w\mid w\in {\cal  L}_L(s) \},$ for $s\in S,$
where $\min$ is taken considering the lexicographical order on
$\Sigma^\star$. Then, whenever two
mergeable states $s$ and $s'$ were found, we could take $s<s'$ if and
only if $\Psi(s)< \Psi(s')$ (lexicographically).

In the general case, given an injective function $\Psi$, let
$\prec_\Psi$ be an ordering such that $(\forall s, s'\in 
S), s \prec_\Psi s' $ if:

\begin{enumerate}
\item $\Delta(s)<\Delta(s')$, where $<$ is the lexicographical order;  
\item if $\Delta(s) = \Delta(s')$ then $\Psi(s) < \Psi(s')$. 
\end{enumerate}

The algorithm of page~\pageref{phi}, that computes the function $\varphi$
can be adapted for \Adfa{}s by not considering the state  $\pi$,
initializing $i$ with $1$ and considering the order $\prec_\Psi$.
%for the assignment of  state numbers,  becomes:
%The assignment of state numbers, $\varphi$, can be
 %described by the following algorithm:
% \begin{lstlisting}
% $\varphi(\Omega)\leftarrow 0$;$i \leftarrow 1$
% for $l$ in $[0,\rk({\cal A})]$
%   for $s\in R_l$ by  order $\prec_\Psi$
%     $\varphi(s) \leftarrow i$
%     $i \leftarrow i + 1$
% \end{lstlisting}
Consider the \Adfa{} of Figure~\ref{fig:adfa}. Its state ranks are
the following: $R_0=\{s_7,s_8\}$, $R_1=\{s_6\}$, $R_2=\{s_4\}$,
$R_3=\{s_1,s_5\}$, $R_4=\{s_2,s_3\}$ and $R_5=\{s_0\}$.  Regarding
the function $\Psi$ above, the function $\varphi$ is defined by the
following tuples:
$(s_7,1),(s_8,2),(s_6,3),(s_4,4),(s_5,5),(s_1,6),(s_3,7),(s_2,8), (s_0,9)$.
 And, its  string representation is
 {\small$$\scriptsize{[[0,0,0,0][0,0,0,1][0,0,0,1][1,1,2,0] 
[3,1,1,0][0,1,0,0][4,4,0,0][5,0,0,0][5,0,0,0][6,7,8,0]]};$$}
\noindent which is lexicographically ordered within a rank (i.e., respects
condition~\ref{eq:N6}, considering $\leq$ instead of $<$).

As we aim to obtain an exact generator that will
proceed by increasing  rank order, it is convenient that $\Psi(s)$ 
is related to a maximal word of ${\cal L}_L(s)$. To assure that 
in a rank the state representations are lexicographically ordered we
also take into consideration the ranks and the finalities of the states.

Let ${\cal A}=(S\cup \{\Omega\},\Sigma,\delta,s_0,F)$ be an \Adfa,
with $\Sigma$ ordered. For each state $s\in S$, let $\delta^{-1}(s) =
\{(s',\sigma)\mid \delta(s',\sigma)=s\}$, and for $(s',\sigma) \in
\delta^{-1}(s)$ let consider the tuple $\tau=(rk(s'),\sigma,f_{s'})$ with
$f_{s'}=1$, if $s'\in F$ or $0$, otherwise.  We define ${\cal L}_L^{rk}(s)$
to be the set of sequences of these tuples
$\tau_{0}\ldots\tau_{l}$ such that $\sigma_{l}\cdots\sigma_{0}\in {\cal L}_L(s)$.  The
\textit{characteristic word} of $s$, is
 $$\Psi_c(s)=\min\{\tau_0\ldots\tau_l| \tau_0\ldots\tau_l \in {\cal L}_L^{rk}(s)\},$$
 \noindent where $\min$ is taken lexicographically.

\begin{figure}
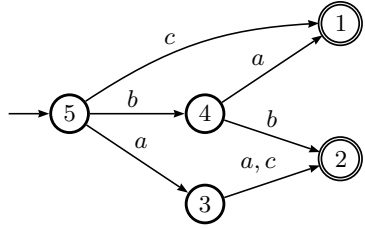

\begin{center}
 \VCDraw{%
        \TinyPicture
        \begin{VCPicture}{(0,-2)(6,3)}
          \State[5]{(0,0)}{S0}
          \Initial{S0}
   %       \State[3]{(3,0)}{S1}
   %       \State[4]{(3,-2)}{S2}
          \State[3]{(3,-2)}{S1}
          \State[4]{(3,0)}{S2}
          \FinalState[2]{(6,-1)}{S3} 
          \FinalState[1]{(6,2)}{S4}
          \EdgeL{S0}{S1}{a}
          \EdgeL{S0}{S2}{b}
          \ArcL{S0}{S4}{c}
          \EdgeL{S1}{S3}{a,c}
          \EdgeL{S2}{S3}{b}
          \EdgeL{S2}{S4}{a}
        \end{VCPicture}
        }
  \end{center}
   \caption{An \Adfa{} which canonical string representation
     considering $\Psi_c$ is: $[[0,0,0,0][0,0,0,1][0,0,0,1][1,0,1,0][1,2,0,0][3,4,2,0]]$.
}
   \label{fig:adfaw}
 \end{figure}

In the example of Figure~\ref{fig:adfaw}, we have $\Psi_c(1)=1a02b0$ and
$\Psi_c(2)=1a02a0$ which shows that the numbers assigned to these states
must be reversed, i.e.,  $\varphi(1)=2$ and $\varphi(2)=1$.

The following three theorems guarantee that this representation is
indeed a canonical representation for \Adfa{}s.

\begin{theorem}
  \label{th:adfa}
  Let ${\cal A}=(S\cup \{\Omega\}, \Sigma,\delta,s_0,F)$ be an \Adfa{}
  with $rk({\cal A})=d$, $n=|S|$ and $k=|\Sigma|$. Let $(s_i)_{i\in
    [0,(k+1)(n+1)[}$, with $s_i\in [0,n[$, be the string
  representation of ${\cal A}$ obtained using the the function
  $\Psi_c$.   Then the conditions \ref{eq:N0}--\ref{eq:N5}
  of Theorem~\ref{th:mcr} are satisfied, together with the following  condition \ref{eq:N6e}:
\begin{gather*}
(\forall l\in[0,d[)(\forall i\in R_l)\, i\prec_{\Psi_c} i+1.  \tag{N6'}\label{eq:N6e}
\end{gather*}

\end{theorem}
\begin{proof}
  Follows from the above considerations.
\end{proof}

\begin{theorem}\label{th:nfadfa}
  Let $(s_i)_{i\in [0,(k+1)(n+1)[}$ with $s_i\in [0,n[$ be a string
  that satisfies conditions~\ref{eq:N0}--\ref{eq:N5} and condition
  \ref{eq:N6e}, then the corresponding automaton is an \Adfa{} with
  $n$ states and an alphabet of $k$ symbols.
\end{theorem}
\begin{proof}
From conditions~\ref{eq:N0}--\ref{eq:N5}, we knew that we could obtain
a trim complete acyclic deterministic finite automaton. The relaxation
of condition~\ref{eq:N6} to condition~\ref{eq:N6e} allows some states
to be mergeable.
\end{proof}
\begin{theorem}\label{th:dfanf}
  Let $(s_i)_{i\in [0,(k+1)(n+1)[}$ and $(s'_i)_{i\in [0,(k+1)(n+1)[}$
  be two distinct strings satisfying
  conditions~\ref{eq:N0}--\ref{eq:N5} and condition \ref{eq:N6e}. Then
  they correspond to distinct \Adfa{}s.
  \end{theorem}
  \begin{proof}
    The proof follows exactly the lines of~Theorem 5 in~Almeida \textit{et
      al.}\cite{almeida07:_exact}, because of condition~\ref{eq:N6e}.
  \end{proof}

\section{Exact generation of \Adfa{}s}
\label{sec:exact}

To generate all the string representations of the \Adfa{}s with $n$
states and $k$ symbols, we will use the same approach described by
Almeida \textit{et al} \cite{almeida07:_exact}, traversing the search
tree, backtracking on its way, to generate all possible
representations. The representations will appear lexicographically. The
conditions to generation are the same but with \ref{eq:N6} replaced by
its relaxed form \ref{eq:N6e}. The satisfaction of the conditions on
the order of equivalent states is too complex to be included in the
generation. When a pair of equivalent states is generated, instead of
renumbering them according to the first word (for some order) that
reaches each state, we proceed with the generation of all the states
in lexicographical order of their $\Delta$ values, and discard the
automata for which the previously stated order is contradicted.

The problem with this strategy is that, with the ``natural''
lexicographical order, the contradiction to the order of two states in
rank $0$ may appear only when generating the last state, i.e., the
\textit{initial state} of the automaton. This is very inconvenient,
because a lot of generating work is going to be discarded and because
of the backtracking strategy, the corresponding search tree is not
pruned as it should. On the other hand, using the order described in
Section~\ref{sec:normal} we can evaluate the possible contradictions
after the complete generation of each rank of states.

The algorithm goes as follows:
\begin{itemize}
\item At the beginning of the generation of each rank, there are two
  data structures:
  \begin{description}
  \item[\texttt{ProbL}] a set of lists of states that are equivalent
    and for which we want to ensure that the characteristic words that reach
    them are in accordance with that order;
  \item[\texttt{Refs}] an empty set of lists of states that, in that
    rank, have transitions to states in some list in \texttt{ProbL}.
  \end{description}
\item Every time two or more states with the same $\Delta$ are generated,
  they are added as a new list to \texttt{ProbL}. 
  \begin{multline*}
      (\Delta(s_1)=\Delta(s_2)=\cdots=\Delta(s_l))\, \wedge\,
      (\varphi(s_1)<\varphi(s_2)< \cdots<\varphi(s_l)) \,\Rightarrow\,\\
      \Rightarrow\, \text{\texttt{ProbL}}\leftarrow\text{\texttt{ProbL}} \cup \{[s_1,s_2,\ldots,s_l]\}.   
    \end{multline*}
  \item Every time a newly generated state has a transition to a state
    present in a list of \texttt{ProbL}, it is added to \texttt{Refs}
    with information about the %respective state of arrival.
state it has a transition to.
  \item When the state generation of a given rank is finished (because
    no more states in that rank can be generated according to rules
    \ref{eq:N0}--\ref{eq:N6e}), each list $R$  in \texttt{Refs} of states
    with transitions to states in a list $L$ in \texttt{ProbL} is examined.
    \begin{itemize}
  \item  For $x\in L$, let  $m(x)=\min\{(\sigma,f_s)\,\mid\,(s,\sigma)\in\delta^{-1}(x)\,\wedge\,s\in
    R\}$, where $f_s$ represents the finality of $s$. 
    If for some pair of states of $L$ a contradiction is found, 
    i.e.,
    \begin{multline*}
      (\exists x_1,x_2\in L)(\exists s_1,s_2\in
      R)(\exists\sigma_1,\sigma_2\in \Sigma)\\
      (\delta(s_1,\sigma_1)=x_1\,\wedge\,\delta(s_2,\sigma_2)=x_2\;\wedge\;m(x_1)<m(x_2)\;
      \wedge\; \varphi(x_1)> \varphi(x_2));
    \end{multline*}
    then the generation of this automaton is aborted and the process
    is continued by backtracking.
  \item For all the non-singleton sublists $M_{(\sigma,f)}$ of states in $L$ such
    that
    $$(\forall x \in M_{(\sigma,f)}) m(x)=(\sigma,f);$$ 
    its elements are removed from $L$, and the list of the states $s$ of
    $R$ such that $(\delta(s,\sigma)\in M_{(\sigma,f)}\,\wedge f_s=f)$,
    with the order induced by $L$, is added to \texttt{ProbL}.
  \item Finally, if 
  \begin{multline*}
      (\forall x_1,x_2\in L)(\exists s_1,s_2\in
      R)(\exists\sigma_1,\sigma_2\in\Sigma)\\
      ((\delta(s_1,\sigma_1)=x_1\,\wedge\, \delta(s_2,\sigma_2)=x_2)\,\Rightarrow\,
      (m(x_1)<m(x_2)\,\Rightarrow\varphi(x_1)<\varphi(x_2)));
    \end{multline*}
    then all the states in $L$ that are the image of a transition from a
    state in $R$ are removed from $L$, and $R$ is removed from \texttt{Refs}.
  \item All empty or singleton lists are removed from \texttt{ProbL}.
  \item Before the generation of a new rank is started, \texttt{Refs}
    is emptied.
\end{itemize}
\end{itemize}

The correctness of this algorithm follows from the considerations in
Section~\ref{sec:normal}.

\subsection{Some experimental results}
\label{sec:result}

In Table~\ref{tab:numbers} the number of \Madfa{}s and \Adfa{}s for
some small values of $n$ and $k$ are summarized. We observe that
almost all \Adfa{}s are \Madfa{}s. Several performance times are also
presented. For the enumeration of \Adfa{}s and \Madfa{}s instead of
the exact generators, we also generate initially-connected
deterministic automata (\Icdfa{}s), using the method presented in
Almeida \textit{et
  al.}~\cite{almeida07:_enumer_gener_strin_autom_repres}, and then
test for acyclicity, trimness and possibly, for minimality. But the
number of \Idfa{}s grows much faster then the number of \Adfa{}s (or
\Madfa{}s), so the generate-test-reject method is not feasible.  In
column \textbf{Time B} of Table~\ref{tab:numbers} we present the
running times obtained by this method (for small values of $n$ and
$k$).  In column \textbf{Time A} of Table~\ref{tab:numbers} we present
the running times obtained by the exact generation methods.

\begin{table}
  \centering
 {\small   \begin{tabular}{c}
 \begin{tabular}[h]{|c||r|r|r|c||r|r|r|}
    \hline
&\multicolumn{7}{|c|}{$k=2$}\\\hline
    \textbf{$n$} &\multicolumn{1}{|c|}{\textbf{MADFA}} & \multicolumn{1}{|c|}{\textbf{Time A (s)}} &
\multicolumn{1}{|c|}{\textbf{Time B (s)}} & \multicolumn{1}{|c||}{\textbf{A/B}} &\multicolumn{1}{|c|}{\textbf{ADFA}} &
\multicolumn{1}{|c|}{\textbf{Time A (s)}} &\multicolumn{1}{|c|}{\textbf{Time B (s)}}
    \\
    \hline
    3 & 60 & 0 & 0 & 1 & 62 & 0 & 0 \\
    \hline
    4 & 900 & 0 & 0 & 1 & 964 & 0.14 & 0 \\
    \hline
    5 & 18480 & 0.1 & 0.1 & 1 & 20424 & 3.65 & 0.2 \\
    \hline
    6 & 487560 & 3 & 16 & 0.18 & 553472 & 110.92 & 13.92 \\
    \hline
    7 & 15824880 & 116 & 687 & 0.16 & 18384552 &  & 736.89 \\
    \hline
    8 & 612504240 & 4742 & 35774 & 0.13 & 726133776 &  & 20284.14 \\
    \hline
    9 & 27619664640 & 224243 &2345124 & 0.09 &  &  &  \\
    \hline
  \end{tabular}\\
 \begin{tabular}[h]{|c||r|r|r|c||r|r|r|}
    \hline\hline
&\multicolumn{7}{|c|}{$k=3$}\\\hline
    \textbf{$n$} &\multicolumn{1}{|c|}{\textbf{MADFA}} & \multicolumn{1}{|c|}{\textbf{Time A (s)}} &
\multicolumn{1}{|c|}{\textbf{Time B (s)}} & \multicolumn{1}{|c||}{\textbf{A/B}} &\multicolumn{1}{|c|}{\textbf{ADFA}} &
\multicolumn{1}{|c|}{\textbf{Time A (s)}} &\multicolumn{1}{|c|}{\textbf{Time B (s)}}\\

    \hline
    2 & 14 & 0 & 0 & 1 & 14 & 0 & 0 \\
    \hline
    3 & 532 & 0.1 & 0.1 & 1 & 544 & 0.07 & 0 \\
    
    \hline
    4 & \ \ \ \ \ \ \ \ 42644 & 0.22 & 0.58 & 0.38 &\  \ \ \ \ \
    44290 & 7.59 & 6.16 \\

\hline
    5 & \ \  \ 6011320 & 43 & 3340 & 0.01 &\  \ \ \ 6306672 &  & 3142
 \\
    \hline 
   6 & \ 1330452032 & 11501 & 2431307 &0.005  &\  \ \ &  & 
 \\
    \hline
  \end{tabular}\\
\begin{tabular}[h]{|c||r|r|r|c||r|r|r|}
    \hline\hline
&\multicolumn{7}{|c|}{$k=4$}\\\hline
    \textbf{$n$} &\multicolumn{1}{|c|}{\textbf{MADFA}} & \multicolumn{1}{|c|}{\textbf{Time A (s)}} &
\multicolumn{1}{|c|}{\textbf{Time B (s)}} & \multicolumn{1}{|c||}{\textbf{A/B}} &\multicolumn{1}{|c|}{\textbf{ADFA}} &
\multicolumn{1}{|c|}{\textbf{Time A (s)}} &\multicolumn{1}{|c|}{\textbf{Time B (s)}}\\
    \hline
    2 & 30 & 0 & 0 & 1 & 30 & 0 & 0 \\
    \hline
    3 & 3900 &0.2 & 1.6  & 0.13 & 3950 & 0.51 &1.55  \\
    \hline
    4 &  \ \ \ \ \ 1460700 & 7.7 & 5549 & 0.001 &\  \ \ \ 1488120 &
    236& 5326  \\\hline
  \end{tabular}
    \\  
\begin{tabular}[h]{|c||r|r|r|c||r|r|r|}
    \hline\hline
&\multicolumn{7}{|c|}{$k=5$}\\\hline
    \textbf{$n$} &\multicolumn{1}{|c|}{\textbf{MADFA}} & \multicolumn{1}{|c|}{\textbf{Time A (s)}} &
\multicolumn{1}{|c|}{\textbf{Time B (s)}} & \multicolumn{1}{|c||}{\textbf{A/B}} &\multicolumn{1}{|c|}{\textbf{ADFA}} &
\multicolumn{1}{|c|}{\textbf{Time A (s)}} &\multicolumn{1}{|c|}{\textbf{Time B (s)}}\\
    \hline
    2 & 62 & 0 & 0 & 1 & 62 & 0 & 0 \\
    \hline
    3 & 26164 &0.121 & 2.396  &0.05  & 26344 & 3.4 & 116  \\
    \hline
    4 &  \ \ \ 43023908 &   213& & &\  \ \ 43411218 &6805& 4872111  \\
    \hline
  \end{tabular}\\\\
    \end{tabular}
  }
\caption{Number of \Madfa{}s and \Adfa{}s for small values of $n$ and
  $k$. Performance times for its generation: exact (\textbf{A}) and
  with a test-rejection pass (\texttt{B}).}
\label{tab:numbers}
\end{table}

Considering only the performance times for \Madfa{}s and $k=2$, we
obtained  a curve fitting for both methods: for the exact generation
a function $f(n)= e^{3.66n-20.76}$ and for the test-reject a
function $g(n)=  e^{4.21n-23.0}$, which gives $g(n)/f(n)=e^{0.55n-2.24}$.

As for the performance values we should only consider their order of
magnitude as they were obtained using different CPUs and programs
implemented in different programming languages. Both performance times
\textbf{B}, were obtained using a \texttt{C} implementation and
running on a AMD Athlon 64 at 2.5GHz.  Performance times \textbf{A}
were obtainned using a \texttt{C++} implementation and running on a
Intel\textsuperscript{\textregistered}
Xeon\textsuperscript{\textregistered} 5140 at 2.33GHz, and a
\texttt{Python} implementation running on a AMD Athlon 64 at 2.5GHz,
respectively for \Madfa{}s and for \Adfa{}s (in general the
\texttt{C++} implementation for \Madfa{} is two times faster than the
correspondent \texttt{Python} implementation).
 
It is reasonable that for (very) small
values of $n$ the test-reject method is faster, as the pruning
of non legal \Adfa{}s is a relatively costly operation. But because of the much faster
growing of the number of \Icdfa{}s (when compared with the  number of
\Adfa{}s), that will not happen for larger $n$s.
\vspace{-0.5cm}
\section{Counting \Adfa{}s for $n$ and $k$}
\label{sec:count}
Let $A_k(n)$ be the number of \Adfa{}s with $n$ states over an
alphabet of $k$ symbols and let $M_k(n)$ be the corresponding number of
\Madfa{}s. In Almeida \textit{et al.}~\cite{almeida07:_exact}, the
values of $M_k(n)$ were determined for $n\in [1,5]$. The same kind of
results can be obtained for $A_k(n)$. The values of $A_k(n)$ for
small values of $n$ can be determined by considering the possible
distribution of states by ranks and the number of \textit{dangling}
states that are targets of transitions from a state of a previous rank,
for the first time.  Using the Principle of Inclusion and Exclusion we
have:
\begin{align*}
A_k(2)=&M_k(2)=2(2^k-1).\\
A_k(3) =& M_k(3)+ (3^k-2^{k+1}+1)= 2^2(3^k-2^k)(2^k-1)+ (3^k-2^{k+1}+1).\\
 A_k(4) =& 2^3(4^k-3^k)(3^k-2^k)(2^k-1) + 2^2(4^k-3^k2+2^k)(2^k-1)^2
 \\
+& 2(4^k-3^k)(3^k-2^k2+1)+(4^k- 3^k3+2^k3-1)/3.
\end{align*}
For $n=5$ there are already 12 configurations to be considered.
For values of $n\in[2,5]$, $\lim_{k\rightarrow
  \infty}M_k(n)/A_k(n)=1$. We note that this behaviour is also
observed (experimentally) in the case of arbitrary \Icdfa's.

\section{Conclusions}
\label{sec:con}
A canonical representation for minimal acyclic deterministic finite
automata was extended to allow equivalent states, and thus uniquely
represent trim acyclic deterministic finite automata.  A method for
the exact generation of \Madfa{}s was extended to allow the generation
of equivalent states, while still avoiding the multiple generation of
non-isomorphic automata. More experimental tests must be carried on in
order to see what really is the overhead of pruning non-legal
equivalent states.
 \section{Acknowledgements}
 \label{sec:ack}
We thank the anonymous referees for their
comments that helped to improve this paper.


\begin{thebibliography}{10}

\bibitem{almeida07:_enumer_gener_strin_autom_repres}
M.~Almeida, N.~Moreira, and R.~Reis.
\newblock Enumeration and generation with a string automata representation.
\newblock {\em Theoretical Computer Science}, 387(2):93--102, 2007.

\bibitem{almeida07:_exact}
M.~Almeida, N.~Moreira, and R.~Reis.
\newblock Exact generation of minimal acyclic deterministic finite automata.
\newblock In {\em Workshop on Descriptional Complexity of Formal Systems
  (DCFS07)}, pages 57--68, High Tatras, Slovakia, 20-22/07 2007.

\bibitem{bassino07:_theor_comput_scien}
F.~Bassino and C.~Nicaud.
\newblock Enumeration and random generation of accessible automata.
\newblock {\em Theoretical Computer Science}, 381(1-3):86--104, 2007.

\bibitem{callan07:_stirl_cycle_number_count_unlal}
D.~Callan.
\newblock A determinant of {S}tirling cycle numbers counts unlabeled acyclic
  single-source automata.
\newblock Department of Statistics, University of Wisconsin-Madison, 2007.

\bibitem{champarnaud:_random_gener_dfas}
J.-M. Champarnaud and T.~Parantho\"en.
\newblock Random generation of {DFA}s.
\newblock {\em Theoretical Computer Science}, 330(2):221--235, 2005.

\bibitem{domaratzki06:_enumer_formal_languag}
M.~Domaratzki.
\newblock Enumeration of formal languages.
\newblock {\em Bull. EATCS}, 89(113-133), June 2006.
\newblock 2006.

\bibitem{domaratzki02}
M.~Domaratzki, D.~Kisman, and J.~Shallit.
\newblock On the number of distinct languages accepted by finite automata with
  n states.
\newblock {\em J. of Automata, Languages and Combinatorics}, 7(4):469--486,
  2002.

\bibitem{hopcroft00_c:_introd_autom_theor_languag_comput}
J.~Hopcroft, R.~Motwani, and J.~D. Ullman.
\newblock {\em Introduction to Automata Theory, Languages and Computation}.
\newblock Addison Wesley, 2000.

\bibitem{liskovets69}
V.~A. Liskovets.
\newblock The number of initially connected automata.
\newblock {\em Kibernetika}, 3:16--19, 1969.
\newblock (in Russian; Engl. transl: {\em Cybernetics}, 4 (1969), 259-262).

\bibitem{liskovets06:_exact}
V.~A. Liskovets.
\newblock Exact enumeration of acyclic deterministic automata.
\newblock {\em Discrete Applied Mathematics}, 154(3):537--551, March 2006.

\bibitem{lothaire05:_algor_words}
M.~Lothaire.
\newblock Algorithms on words.
\newblock In {\em Applied Combinatorics on Words}, chapter~1. Cambridge
  University Press, 2005.

\bibitem{nicaud00}
C.~Nicaud.
\newblock {\em Étude du comportement en moyenne des automates finis et des
  langages rationnels}.
\newblock PhD thesis, Université de Paris 7, 2000.

\bibitem{revuz92:_minim_of_acycl_deter_autom}
D.~Revuz.
\newblock Minimisation of acyclic deterministic automata in linear time.
\newblock {\em Theor. Comp. Sci.}, 92(1):181--189, 1992.

\bibitem{robinson85:_count}
R.~W. Robinson.
\newblock Counting strongly connected finite automata.
\newblock In {\em Graph Theory with Applications to Algorithms and Computer
  Science}, pages 671--685. Wiley, 1985.

\bibitem{yu97:_handb_formal_languag}
S.~Yu.
\newblock Regular languages.
\newblock In G.~Rozenberg and A.~Salomaa, editors, {\em Handbook of {F}ormal
  {L}anguages}, volume~1. Springer-Verlag, 1997.

\end{thebibliography}
\end{document}